\documentclass[12pt, hidelinks]{article}
\usepackage{diagbox}
\usepackage[utf8]{inputenc}
\usepackage[T1]{fontenc}

\usepackage{arxiv}

\usepackage{microtype}
\usepackage{graphicx}
\usepackage{subfigure}
\usepackage{booktabs} 

\usepackage{hyperref}
\usepackage{fullpage}
\usepackage{enumerate}
\usepackage{xcolor}
\usepackage{bbm}

\usepackage{algorithm}
\usepackage{algorithmic}

\usepackage{amsmath}
\usepackage{amssymb}
\usepackage{mathtools}
\usepackage{amsthm}

\usepackage[capitalize,noabbrev]{cleveref}

\theoremstyle{plain}
\newtheorem{theorem}{Theorem}[section]

\newtheorem{lemma}[theorem]{Lemma}
\newtheorem{example}[theorem]{Example}
\newtheorem{corollary}[theorem]{Corollary}
\theoremstyle{definition}
\newtheorem{definition}[theorem]{Definition}

\theoremstyle{remark}

\title{Selling Privacy in Blockchain Transactions}
\date{}

\author{
 Georgios Chionas \\
   School of Computer Science and Informatics \\
  University of Liverpool\\
  \texttt{g.chionas@liverpool.ac.uk} \\
   \And
 Olga Gorelkina \\
  Moroccan Center for Game Theory \\
  University Mohammed VI Polytechnic \\
  \texttt{olga.gorelkina@um6p.ma} \\
  \And
 Piotr Krysta\thanks{P. Krysta is also affiliated with University of Liverpool} \\
   School of Computer and Cyber Sciences \\
  Augusta University\\
  \texttt{pkrysta@augusta.edu} \\
\And
 Rida Laraki \\
 Moroccan Center for Game Theory \\
  University Mohammed VI Polytechnic \\
  \texttt{rida.laraki@um6p.ma} \\    
}

\begin{document}
\newcommand{\bfx}{\mathbf{x}}
\newcommand{\bfp}{\mathbf{p}}
\newcommand{\bfq}{\mathbf{q}}
\newcommand{\tfm}{(\allocs,\prices,\burn)}
\newcommand{\fpa}{(\allocs^f,\prices^f,\burns^f)}

\newcommand\fnsep{\textsuperscript{,}}

\newcommand{\bid}{b}
\newcommand{\bids}{{\mathbf \bid}}
\newcommand{\bidsmi}{{\mathbf \bid}_{-i}}
\newcommand{\bidt}[1][t]{{\bid_{#1}}}

\newcommand{\costs}{{\mathbf{c}}}

\newcommand{\epnebid}{b^*}

\newcommand{\val}{v}
\newcommand{\vals}{{\mathbf \val}}
\newcommand{\valsmi}{{\mathbf \val}_{-i}}
\newcommand{\valt}[1][t]{{\val_{#1}}}
\newcommand{\vmin}{v_{\text{min}}}
\newcommand{\vmax}{v_{\text{max}}}

\newcommand{\alloc}{x}
\newcommand{\allocs}{{\mathbf \alloc}}
\newcommand{\yallocs}{{\mathbf{y}}}

\newcommand{\allocsmi}{\allocs_{-i}}
\newcommand{\alloct}[1][t]{\alloc_{#1}}

\newcommand{\price}{p}
\newcommand{\prices}{{\mathbf \price}}
\newcommand{\pricet}[1][t]{\price_{#1}}

\newcommand{\burn}{q}
\newcommand{\burns}{{\mathbf \burn}}
\newcommand{\burnt}[1][t]{\burn_{#1}}

\newcommand{\cost}{\gamma}
\newcommand{\oca}{(\bids,\allocs,\bm{\tau})}

\newcommand{\nine}{*}
\newcommand{\ninetfm}{(\allocs^{\nine},\prices^{\nine},\burns^{\nine})}

\newcommand{\history}{\mathbf{H}}

\newcommand{\blocks}{\mathcal{B}}
\newcommand{\blockset}{\mathcal{B}}
\newcommand{\vbp}{v_{BP}}
\newcommand{\vbphat}{\hat{v}_{BP}}
\newcommand{\basefee}{r}

\newcommand{\Rplus}{\mathbb{R}_{\ge 0}}

\newcommand{\ve}{(v, \epsilon)}
\newcommand{\vmean}{\Bar{v}(\e)}

\newcommand{\mech}{(\mathbf{x(\mathbf{b}}), \mathbf{p(\mathbf{b}}))}

\newcommand{\kda}{k-Dutch auction }
\newcommand{\revda}{\text{REV-k-DA}}
\newcommand{\Rev}{\text{Rev}}
\newcommand{\Sur}{\text{Surplus}}

\newcommand{\OPT}{\text{OPT}}
\newcommand{\ALG}{\text{ALG}}
\newcommand{\SW}{\text{SW}}

\newcommand{\bphi}{\Tilde{\phi}}

\newcommand{\htau}{\hat{\tau}}

\newcommand{\edp}{$\e$-differential privacy }
\newcommand{\eddp}{$(\e, \delta)$-differential privacy }
\newcommand{\edpm}{$\e$-differential private mechanism}
\newcommand{\e}{\varepsilon}
\newcommand{\la}{\lambda}
\newcommand{\La}{\Lambda}

\newcommand{\EXP}{\text{EXP}}
\newcommand{\an}{$\alpha n$}
\newcommand{\uout}{$u^{\text{out}}$}
\newcommand{\uinf}{$u^{\text{inf}}$}

\newcommand{\ik}{i^{\dagger}}
\newcommand{\jk}{j^{\dagger}}

\newcommand{\takeleave}{take-it-or-leave-it}

\newcommand{\XA}{X^{\mathbb{A}}}

\newcommand{\Search}{\text{S}}
\newcommand{\User}{\text{U}}

\newcommand{\E}{\mathbb{E}}

\maketitle
\begin{abstract} 
We study methods to enhance statistical privacy in blockchain transactions. We analyze economic mechanisms for privacy-aware transaction owners whose utility depends not only on the outcome of the mechanism but also negatively on the exposure of their economic preferences. First, we consider an order flow auction, where a user auctions off to specialized agents, called searchers, the right to execute her transaction while maintaining a degree of privacy. We examine how the degree of privacy affects the revenue of the auction and, broadly, the net utility of the privacy-aware user. In this new setting, we characterize the optimal auction, which is a sealed-bid auction. Subsequently, we analyze a variant of a Dutch auction in which the user gradually decreases the price and the degree of privacy until the transaction is sold. We compare the revenue of this auction to that of the optimal one as a function of the number of communication rounds. Then, we introduce a two-sided market - a privacy marketplace - with multiple users selling their transactions under their privacy preferences to multiple searchers. We propose a posted-price mechanism for the two-sided market that guarantees constant approximation of the optimal social welfare while maintaining incentive compatibility (from both sides of the market) and budget balance. 
This work builds on the emerging literature on privacy-preserving mechanism design, integrating statistical privacy guarantees into economic protocols to capture the impact of information leakage on blockchain users' utility.
\end{abstract}

\section{Introduction}
\label{sec: introduction}
Smart contract blockchains such as Ethereum \cite{Buterin13} enable one to construct flexible financial applications on top of blockchains and led to the emergence of \textit{Decentralized Finance} (DeFi). Users interact with those applications by submitting their transactions. These submitted transactions are then aggregated into blocks by agents called block producers. Even though blockchains are considered to be decentralized, temporarily, they become centralized, since in each block, the chosen block producer has unilateral power to select and order the transactions of the newly created block. This monopolistic power in conjunction with the public data of the pending transactions (public mempool) enables agents involved in the block production process to extract profit, usually at the expense of users. This phenomenon has been termed MEV \cite{Daian19}. In order to mitigate the negative externalities that come from the monopolistic power of the block producers and also decouple the role of block building and block validation, Ethereum has followed the so-called Proposer Builder Separation (PBS) roadmap \cite{Flashbots_boost} and thus currently, the block production supply chain consists of proposers, builders and searchers. In our work, we focus on the second source of value extraction, which is the public content of the pending transactions. There has been a huge amount of literature on how to safeguard users' data and their financial activity by leveraging cryptographic techniques. For example, see this survey \cite{Baldimtsi24PrivacyTransactions}. 
One may consider two extreme points. In one extreme, which is the current status of most layer-one (L1) mempools, all transactions are public and hence specialized agents (searchers) who monitor the pending transactions can take advantage of them in order to extract profit at the expense of users who submitted them. The other extreme is fully private mempools, which do result in inefficient markets and applications, since any kind of financial application requires some data to be public. 
In this work, we seek to study mechanisms that lie in the middle of the privacy spectrum, where users can internalize the value of their transaction while also maintaining some degree of privacy in their transactions. Nonetheless, we only scratched the surface of what these mechanisms might be able to offer, for example, optimizing for privacy, efficiency, user surplus, and other key design objectives.

\subsection{Related Work}
\label{sec:related_work}
The pre-MEV literature on mechanism design mainly considered two types of agents; the users that express their economic preferences through submitting transactions combined with the associated transaction fees, and the block producers, whose utility was captured by the block reward and the collected fees \cite{Roughgarden20TFM}. As the block production supply chain evolved, new types of agents have emerged, which have rendered the design of efficient mechanisms more complicated. This, for example, is demonstrated by the impossibility result of  \cite{Bahrani24Post-MEV}. This impossibility result formalizes the necessity for designing mechanisms that allow users to better express their economic preferences, such as their privacy preferences in their transactions.  
In this work, we leverage ideas from Differential Privacy \cite{Dwork2014differential} that provides a mathematical framework to reason about the tradeoff between utility and privacy. Differential privacy has also been used as a solution concept in mechanism design \cite{McSherry07DP_MD} and subsequently, incentive compatible mechanisms for privacy-aware agents have been proposed \cite{Ghosh11SellingPrivacy,Nissim12Privacy}. In the blockchain literature, to the best of our knowledge, there are only two works \cite{Chitra22differential,Goyal23personalized} that use differential privacy in order to add privacy in users' trades on Constant Function Market Makers (CFMMs) \cite{Angeris20CFMMs}.

\paragraph*{Dutch Auctions.} Also known as descending price auctions, Dutch auctions are auctions in which an item is listed at a high price, which gradually decreases over time until a bidder accepts. Dutch auctions are used in a variety of applications in decentralized finance due to their simple design, efficiency, and low communication complexity. To name a few examples, Dutch auctions are used for liquidations in MakerDAO \cite{makerdao2020protocol}, for trading protocols such as UniswapX \cite{Adams2023uniswapx}, and have also been proposed to set transaction fees in EIP-2593 \cite{Finlay2020eip2593}. 
From a mechanism design perspective, Dutch auctions can be viewed as posted price mechanisms enhanced with competition. This competition, in turn, helps with increasing the generated revenue (compared to posted price mechanisms), while also making the auction practical and light in communication between the auctioneer and the bidders. That is because the communication is limited to the announced prices by the auctioneer and the winning bid, instead of aggregating all the bidders' bids.
Moreover, as the number of price levels increases, with properly chosen price levels, the revenue of the Dutch auction converges to the optimal one. We use tools from the rich literature on prophet inequalities and posted-price mechanisms, cf. \cite{Ehsani18prophet,Chawla2010SPP,Lucier17Prophets,Alaei19OptimalvsAnonymous}, to analyze a variant of Dutch auction in \Cref{sec:dutch_auction}.

Our work lies in the unexplored field of MEV redistribution mechanisms. These mechanisms enable users to capture the value their transactions create in the block production process.
Examples of such mechanisms are the proposal of MEV-burn \cite{Drake2023mev}, which is an add-on to an enshrined PBS framework. We remind that burning is implicitly a redistribution to the token holders. Another mechanism that lies in this category while also closely related to our work is the product of Flashbots, MEV-Share \cite{Miller23}. In this protocol, users can customize the visibility of their transaction data to searchers, but most notably specify the distribution of the searchers' payment. Lastly, \cite{Chionas24MEV-Sharing} study the dynamical behavior of a two-sided market, formed by users who create MEV and agents involved in the block production process who extract it. 

\subsection{Our Contributions}
As stated above, the ability of the agents involved in the block production process to extract value comes from monitoring the public content of pending users' transactions.
Our main contribution is to define and study economic mechanisms and analyze the welfare of \textit{privacy-aware} users who can configure the degree of privacy in their transactions.
Intuitively, we may view a user's utility for having her transaction executed as the sum of two components; a non-negative component that results from the new state of the blockchain and a negative component that results from the leakage of information about their economic preferences. We consider two auction-theoretic settings. First, in \cref{sec:privacy_auction}, we study the setting where a user auctions off her transaction to searchers while maintaining a degree of privacy in her transaction. The degree of privacy is captured by a differential privacy algorithm and the differential privacy parameter $\e$ which corresponds to the amount of noise added to the transaction. Naturally, the transaction becomes more valuable to searchers the lower the degree of privacy is. However, the user who values privacy suffers a higher \textit{privacy cost} as the degree of privacy drops. We characterize the optimal auction in this setting; that is, the \textit{truthful} auction that maximizes the user's net utility, which is defined as the revenue of the auction minus the privacy cost. Then, in \cref{sec:dutch_auction}, we describe an alternative practical mechanism which is a variant of the Dutch auction, whereby the user starts by offering her transaction at a high price with a high degree of privacy and gradually lowers the price while also reducing the degree of privacy. By carefully selecting the price tailored to the degree of privacy in each round, we parameterize the user's net utility as a function of the number of rounds of the auction, specifically, achieving a $1 - \frac{1}{e^{\ell}}$ approximation of the optimal user's net utility, with $\ell$ being the number of rounds. Subsequently, in \cref{sec:privacy_marketplace}, we extend this idea to a two-sided market - a privacy marketplace - where multiple users sell their transactions under their privacy preferences and unit-demand searchers buy them. In order for our two-sided market to satisfy desired properties, such as high welfare, incentive compatibility, and budget balance, we extend known techniques from combinatorial auctions in one-sided markets. Ultimately, we design a mechanism that offers searchers a take-it-or-leave-it price-privacy tuple for each transaction and then, the unit-demand searchers arrive sequentially and buy their utility maximizing transaction.
The mechanism results in a constant approximation of the optimal welfare, while satisfying the aforementioned properties.

\section{Setup}
We start by formalizing the privacy and the revenue maximization framework that will be used in the auctions described in the subsequent sections.

\paragraph*{Differential Privacy.}
In order to capture the tradeoff between utility and privacy, we employ the  framework of differential privacy \cite{Dwork2014differential}. The promise of differential privacy is that when an algorithm releases information from a dataset, one by just observing the output of the algorithm, cannot infer much about individual elements of the dataset. In particular, differential privacy states that, for two neighboring datasets, that is, datasets that differ in at most one element, the output from the mechanism satisfies the following property.

\begin{definition}[Differential Privacy]
A randomized algorithm $A$ is $\e$-differentially private if for all neighboring datasets $S,S'$ and all output sets $B$, it holds
\end{definition}
\begin{equation*}
    \Pr[A(S) \in B] \le \text{exp}(\e) \cdot \Pr[A(S') \in B] 
\end{equation*}
Here, $\e$ is the privacy parameter, whereby small values of $\e$ correspond to stronger privacy, while large values of $\e$ correspond to weaker privacy guarantees. Without loss of generality, we will follow the general convention in the differential privacy community by assuming that $\e \in [0,1]$.
Intuitively, each \edpm $ $ defines a \textit{Pareto} curve that trades off utility for privacy. In the context of blockchain applications, \cite{Chitra22differential} design differential private mechanisms tailored to CFMMs. The analogy of applying differential private mechanisms in CFMMs is that they induce a trade-off between price impact and privacy. In order to achieve differential private algorithms in CFMMs, the key idea in \cite{Chitra22differential} is to split a trade into small sub-trades and subsequently randomly permute the execution order of the sub-trades. In fact, randomly permuting the set of trades to be executed is one simple way to introduce entropy to the CFMM.
Nevertheless, we note that the main focus of this paper is the economic analysis of the auctions defined in the subsequent sections rather than the design of differential private mechanisms. Thus, we can use the algorithms of \cite{Chitra22differential} as building blocks when the users' transactions interact with CFMMs. 

\paragraph*{Optimal One-dimensional Mechanism Design.}
Consider a single item auction, with $n$ bidders. A bidder $i$ has a private value $v_i \in \mathbb{R}_{\ge0} $ for the item and a quasilinear utility, that is, if she wins, she receives the item, she pays $p_i$ and her utility is $v_i - p_i$, otherwise, her utility is 0. We assume a Bayesian model, where $v_i$ is drawn independently from a known distribution with cumulative distribution function (cdf) $F$ and probability density function (pdf) $f$. 
A bidding strategy $b_i$ for a bidder $i$ is a function that takes as input her private valuation $v_i$ and outputs an action $b_i(v_i )$, that is, a bid given $v_i$. Consider a bidding profile $\mathbf{b} = (b_1, \dots , b_n)$. Then, a mechanism is the tuple $\mech$, where $x_i(\mathbf{b})$ denotes the probability that bidder $i$ gets the item and $p_i(\mathbf{b})$ is her payment.
In order to formally describe the optimal auction, which is the auction that maximizes the auctioneer's expected revenue, we need to define the concept of virtual value function which was described in Myerson's seminal work \cite{Myerson81Optimal}.
\begin{definition}
For a valuation $v_i$ drawn from $F$, the virtual value function of agent $i$ is given by
\begin{equation*}
\phi(v_i) =  v_i - \frac{1- F(v_i)}{f(v_i)}
\end{equation*}
\end{definition}
Myerson's optimal mechanism is based on the following critical observation. 
\begin{theorem}[Myerson's Theorem~\cite{Myerson81Optimal}]
Consider a strategyproof auction that awards the item to buyer $i$ with probability $x_i(\mathbf{v})$ and charges $p_i(\mathbf{v})$ on bids $\mathbf{v}$. Then, the expected revenue is
\begin{equation}
\mathbb{E}_{\mathbf{v}}\left[{\sum_{i = 1}^n p_i(\mathbf v)} \right ] = \mathbb{E}_{\mathbf v}\left [{\sum_{i = 1}^n \varphi(v_i) \cdot x_i(\mathbf v)} \right ] 
\end{equation}

\end{theorem}
The right-hand side is called the expected {\em virtual welfare}. For cases where $F$ is regular, i.e., the virtual value function $\varphi$ is non-decreasing, the optimal auction maximizes expected virtual welfare.

\section{Selling Privacy at Auction}
\label{sec:privacy_auction}
In this section, we study how a user can auction off her transaction to searchers while maintaining a degree of privacy in her transaction. In order to model this auction, the user employs a differential privacy algorithm (such as the ones described in \cite{Chitra22differential} if her transaction interacts with a CFMM) and the differential privacy parameter $\e$ which calibrates the degree of privacy.
There are $n$ searchers that seek to execute the user's transaction. We consider a Bayesian model. Each searcher $i$'s valuation $X$ for the user's transaction is drawn from a known distribution with a marginal distribution for $v$ conditioned on a given privacy parameter $\e$, denoted by $F_{\e}(v) =  P[X \le v|\e]$.
We naturally assume that a lower degree of privacy (i.e., higher values of $\e$) increases searchers' valuations since they learn more about the user's transaction, 
and we model this by assuming hazard rate order; that is, for $\e_1\le\e_2$ and any $v$, it holds $\frac{f_{\e_1}(v)}{1- F_{\e_1}(v)} \ge \frac{f_{\e_2}(v)}{1- F_{\e_2}(v)}$. 
On the other hand, a low degree of privacy yields an increase in the user's \textit{disutility} for having information leaked about her transaction. In our model, this is captured by the user's privacy cost. We define the privacy cost as the cost function $c : [0,1] \to \Rplus$, where $c(\e)$ captures the user's disutility for having her transaction used in a $\e$-differentially private manner.
We consider convex, non-decreasing (in $\e$) cost functions.\footnote{We also assume the privacy cost $c(\e)$ can be measured in monetary terms and thus, it is denominated in the same unit as the revenue of the auction.}
Our goal is to study the mechanisms that maximize or approximate the net utility of the privacy-aware user:
\begin{equation}
\label{eq:utility_privacy_aware}
    u_0(\e, n) = \text{Rev}(\e, n) - c(\e),
\end{equation}
where $\text{Rev}(\e, n)$ is the expected revenue of the auction with $n$ searchers by selling the transaction in an $\e$-differentially private manner. Intuitively, we may think of the privacy parameter $\e$ from the bidders' perspective as selling an item at several quality levels. From the user's perspective, who acts as the auctioneer, an increase in $\e$ yields a higher privacy cost. In fact, the privacy cost can be thought of as the cost of offering an item at a higher quality level. In our case, the user trades off revenue against the cost of leaking privacy in her transaction. 


\paragraph{Discrete Privacy Levels.} From now on, for our modeling purposes, we will consider $\ell$ discrete privacy levels $\e_1 < \e_2 < \dots < \e_{\ell}$ in the interval $[0,1]$ instead of the continuous interval. We will reserve the usage of letter $k$ to denote a privacy level $k \in [\ell]$.
Moreover, in order to avoid burdensome notation, we will denote the marginal valuation distribution $F(v| \e= \e_k)$ as $F_k(v)$ and similarly, we will write $c_k$ to denote the privacy cost the user suffers at privacy level $\e_k$, where $k \in [\ell]$.

\paragraph*{Optimal Mechanism.}
\label{par:optimal_mechanism}
Following \cite{Myerson81Optimal}, we extend the notion of virtual values in order to incorporate the various privacy levels, privacy costs and hence, for a privacy level $k \in [\ell]$, we define the privacy-enhanced virtual value function $\bphi_k$
\begin{equation}
\label{eq:privacy-virtual-value}
    \bphi_{k}(v) = v - \frac{1 - F_{k}(v)}{f_{k}(v)} - c_k
\end{equation}
\begin{lemma}
Consider any mechanism $\mech$ in which any searcher who bids zero gets a zero payoff. Then, the user's expected net utility under the mechanism $\mech$ is given by:
\begin{equation}
    \mathbb{E}\left[ \sum\limits_{i=1}^n \sum\limits_{k=1}^{\ell} \bphi_{k}(\mathbf{b})  x_{i,k}(\mathbf{b})  \right]
\end{equation}
\end{lemma}

\begin{proof}
The proof follows directly from Lemma 3 in \cite{Myerson81Optimal}, with minor differences due to the inclusion of privacy costs, but we provide the proof for completeness. We considered a fixed privacy level $k$. And now, we start by Myerson's payment formula for searcher $i$:
\begin{equation}
\label{eq:myerson-formula}
    p_i(\vals) = \int_0^{v_i} z \cdot x_{i,k}'(z, \valsmi) dz
\end{equation}
We can now write the expected payment by searcher $i$ for a given value profile $\valsmi$ as:
\begin{equation}
    \mathbb{E}_{v_i \sim F_k}[p_i(\vals)] = \int_0^{v_{max}}p_i(\vals)f_k(v_i)dv_i = \int_0^{v_{max}}\left[\int_0^{v_i}z\cdot x_{i,k}'(z, \valsmi)dz\right]f_k(v_i)dv_i
\end{equation}
Next, we reverse the order of integration, which leads to:
\begin{align*}
 \int_0^{v_{max}}\left[\int_0^{v_i}z\cdot x_{i,k}'(z, \valsmi)dz\right]f_i(v_i)dv_i & =   \int_0^{v_{max}}\left[\int_z^{v_{max}} f_k(v_i)dv_i \right]z\cdot x_{i,k}'(z, \valsmi)dz \\
 & = \int_0^{v_{max}}(1- F_k(z)) \cdot z \cdot x_{i,k}'(z, \valsmi)dz
\end{align*}
Next, we integrate by parts:
\begin{equation*}
\begin{array}{l}
   \int_0^{v_{max}} (1- F_k(z)) \cdot z \cdot x_{i,k}'(z, \valsmi)dz \\
   = ((1-F_k(z)) \cdot z \cdot x_{i,k}(z, \valsmi)) \Big|^{v_{max}}_0  
     - \int_0^{v_{max}} x_{i,k}(z, \valsmi) \cdot \big((1-F_k(z)) - z f_k(z)\big) dz \\
   = \int_0^{v_{max}} \left(z - \frac{1 - F_k(z)}{f_k(z)} \right) 
     x_{i,k}(z, \valsmi) f_k(z) dz = \int_0^{v_{max}} \phi_k(z)  x_{i,k}(z, \valsmi) f_k(z) dz 
     \\ = \mathbb{E}_{v_i \sim F_k} \left[\phi_k(v_i) \cdot x_{i,k}(v) \right]
\end{array}
\end{equation*}
Thus, the user's expected net utility for selling the transaction at privacy level $k$ is
\begin{equation*}
    \mathbb{E}_{v_i \sim F_k}[p_i(\vals) - c_k] = \mathbb{E}_{v_i \sim F_k} \left[\left(\phi_k(v_i) - c_k\right) \cdot x_{i,k}(v) \right] = \mathbb{E}_{v_i \sim F_k} \left[\bphi_k(v) \cdot x_{i,k}(v) \right]
\end{equation*}
Lastly, by applying the linearity of expectations twice, we obtain the desired outcome:
\begin{align*}
\mathbb{E}_{\vals}\left[\sum_{i=1}^n\sum_{k=1}^{\ell}(p_i(\vals) - c_k)\right] & =  \sum_{i=1}^n \sum_{k=1}^{\ell}\mathbb{E}_{\vals}\left[p_i(\vals) - c_k \right] = \sum_{i=1}^n \sum_{k=1}^{\ell} \mathbb{E}_{\vals} \left[\bphi(v_i) \cdot x_i(\vals)\right] \\
& = \mathbb{E}_{\vals}\left[\sum_{i=1}^n \sum_{k=1}^{\ell} \bphi_{k}(\vals) \cdot x_{i,k}(\vals) \right]
\end{align*}

\end{proof}

\paragraph*{Truthful Payment of the Optimal Mechanism.}
As stated in the previous paragraph, the winner in the optimal mechanism is the privacy-enhanced virtual value function maximizer searcher $i^*$, that is, 
\begin{equation*}
    (i^*,k^*) = \arg\max\limits_{(i,k)} \left\{\bphi_{{k}}(v_i)\right\} \, .
\end{equation*}
In order now to design a truthful mechanism, we use the standard notion of critical payments, which means that the winning searcher should pay the minimum amount that she could have bid so that she would remain the winner. To do so, we find the searcher with the second highest privacy-enhanced virtual valuation, and let 
\begin{equation*}
    R = \max\limits_{(i,k)\neq (i^*,k^*)}\left\{\bphi_{{k}}(v_i)\right\} \, .
\end{equation*}
Then, the winner $i^*$ pays the unique price $p$ by solving
\begin{equation*}
\label{eq:threshold_payment}
\bphi_{{k^*}}(p) = R \implies p = \bphi^{-1}_{{k^*}}(R) \, .
\end{equation*}

\section{Dutch Auction}
\label{sec:dutch_auction}
The optimal auction described in \cref{sec:privacy_auction} is a sealed-bid auction, which requires from every searcher to submit his valuation for each privacy level.
We now describe a more practical and light in communication auction that can be viewed as a variant of a Dutch auction. Notably, Dutch auctions have reduced communication complexity compared to sealed-bid auctions, since the communication is limited to the announced prices by the auctioneer and the winning bid, instead of aggregating all the bidders' bids. In general, a user does not know or cannot estimate a priori the value of their transaction to searchers. Thus, we describe an iterative and communication-efficient process, which enables the user to discover and hence internalize the value of their transaction.
In this auction, the user starts by offering her transaction at a high price and with a high degree of privacy, and in each subsequent round, the price and the degree of privacy decrease. The auction is not a standard Dutch auction, but it resembles one because, as the auction unfolds, the transaction becomes more valuable to searchers.
The auction terminates after $\ell$ rounds or when a searcher is willing to buy the transaction at that price and privacy level. In other words, searcher $i$ is willing to buy if her valuation $v_i(\e)$ exceeds the announced price at the current privacy level. We would like to approximate the revenue of the following auction as a function of the number $\ell$ of the prices announced by the auctioneer.

In this new auction format, the user commits to $\ell$ price-privacy levels. In other words, the user chooses $\ell$ tuples of price and privacy level $(p_1, \e_1), (p_2, \e_2), \dots , (p_{\ell}, \e_{\ell})$. We would like to show how well the $\ell$-level Dutch auction approximates the optimal mechanism described in the previous section. Our main question is how well we can approximate the optimal net utility \eqref{eq:utility_privacy_aware} by using at most $\ell$ distinct tuples of price and privacy levels. We now describe the $\ell$-Dutch auction. 

\subsection{Dutch Auction with Privacy Cost}
We will reduce the problem of selecting price-privacy tuples for the $\ell$-DA into a version of Prophet Inequality. A relation that was observed in \cite{Alaei22DescendingDescrete} for the single item auction.

\begin{definition}[Batched Prophet Inequality \cite{Alaei22DescendingDescrete}]
Consider a decision-maker that wants to maximize her expected reward in a sequential game with $\ell$ rounds. She commits to picking $\ell$ thresholds, $\tau_1 > \tau_2 > \dots > \tau_{\ell}$.
There are $n$ rewards $V_1, V_2, \dots, V_n$ drawn from a known distribution $G$. In each round $k$, if all the rewards are less than the threshold $\tau_{k}$, the decision-maker proceeds to the next round; otherwise, she picks uniformly at random one of the rewards that are greater than the threshold $\tau_k$. If no reward passes any of the thresholds until the end of round $\ell$, then the decision-maker's reward is zero.
\end{definition}

We highlight the connection between Batched Prophet Inequality and revenue maximization. Consider now the single item auction with the bidders' valuations being drawn from $F$. Let $\text{REV-OPT}(v)$ be the optimal revenue (in expectation). Let $V_i = \phi(v_i)$ be the reward in the Batched Prophet Inequality; hence $G(x) = F(\phi^{-1}(x))$. 
Thus, 
\begin{equation}
\label{eq:batch-to-rev}
      \OPT = \mathbb{E} \left[\max \left\{\max_{i \in [n]}V_i, 0\right\} \right] =   \mathbb{E} \left[\max \left\{\max_{\substack{i \in [n] \\ k \in [\ell]}}\phi_{k}(v_i), 0\right\} \right] = \text{REV-OPT}(v)
\end{equation}
We will now generalize the process in the setting of \cite{Alaei22DescendingDescrete} by considering the user's privacy cost.

\paragraph*{$\ell$-Dutch Auction ($\ell$-DA).} 
The user commits to $\ell$ price-privacy tuples. The privacy levels are fixed and let them be $\e_1 < \e_2 < \dots < \e_{\ell}$.
The user chooses $\ell$ prices $p_1,p_2, \dots , p_{\ell}$, each corresponding to the respective privacy level, such that the sequence of the privacy-enhanced virtual values \eqref{eq:privacy-virtual-value} is sorted in decreasing order, that is,  $\bphi_{1}(p_1) > \bphi_{2}(p_2)> \dots > \bphi_{{\ell}}(p_{\ell})$. We now describe the auction, where the price-privacy levels are restricted to a finite set $\mathcal{T}=\{(p_1, \e_1), (p_2, \e_2), \dots , (p_{\ell}, \e_{\ell})\}$ of cardinality $\ell$. We ask each searcher to submit a bid $b_k \in  \mathcal{T}$. The winner is the searcher with the highest privacy-enhanced virtual value (breaking ties uniformly at random). And the winning searcher pays their submitted bid. 

\begin{corollary}[Decreasing Prices to Decreasing Privacy Enhanced Virtual Values]
A privacy level sequence $\e_1 < \e_2 < \dots \e_{\ell}$, and a decreasing price sequence $p_1>p_2> \dots > p_{\ell}$ result in a decreasing privacy enhanced virtual value sequence $\bphi_{1}(p_1) > \bphi_{2}(p_2)> \dots > \bphi_{{\ell}}(p_{\ell})$ if 
\begin{equation*}
    c_{k+1} - c_{k} > \sup\limits_{v \sim F_k} \left[ \phi_{{k+1}}(v) - \phi_{{k}}(v) \right], \quad \text{for } k=1,2, \dots , \ell-1
\end{equation*}
\end{corollary}

Intuitively, this means that lowering the degree of privacy (i.e., increasing $\e$) raises the standard virtual function $\phi$ because the valuation distribution improves because of the hazard rate order. In order to keep $\bphi$ falling with the degree of privacy, the incremental privacy cost must dominate the virtual-value increase. In other words, the marginal privacy cost should exceed the marginal increase in the (standard) virtual value.

\paragraph*{Analysis of $\ell$-DA.}
It is straightforward to see that the mechanism is monotone both in price and the privacy level. To be more precise, in Bayes-Nash Equilibrium (BNE), the searcher's bid $b_k = (p_k, \e_k)$ will be an increasing step function with discontinuities at certain thresholds $\htau_1> \htau_2> \dots > \htau_{\ell}$, which we define next. More specifically, for a privacy level $k$, if the searcher's true valuation belongs to the interval $[\htau_k, \htau_{k-1}]$, she will bid $b_k = (p_k, \e_k)$.
Intuitively, the interval $[\htau_k, \htau_{k-1}]$ depends on the number of searchers $n$ and the number of rounds $\ell$. More specifically, as $n$ increases, searchers will bid more aggressively in BNE, while the length of each interval $[\htau_k, \htau_{k-1}]$ decreases as $\ell$ increases. As stated above, the allocation rule is piecewise monotone. In particular, in BNE, the searcher will bid $b_k = (p_k, \e_k)$ if her true valuation belongs to an interval $[\htau_k, \htau_{k-1}]$, and she will $b_{k+1} = (p_{k+1}, \e_{k+1})$ if her true valuation belongs to the interval $[\htau_{k+1}, \htau_{k}]$. Thus, the searcher will be indifferent between bidding $b_k$ and $b_{k+1}$ if her true valuation is $v_k= \htau_k$. By equating the searcher's net utility in these two cases, for $k=1,2, \dots , \ell-1$, we get:
\begin{equation}
\label{eq:indifference}
    P_n \left(\frac{F_{k}(\htau_k)}{F_{{k-1}}(\htau_{k-1})}\right) \cdot (\htau_k - p_k) = \left(\frac{F_{k}(\htau_k)}{F_{{k-1}}(\htau_{k-1})}\right)^{n-1} P_n \left(\frac{F_{{k+1}}(\htau_{k+1})}{F_{{k}}(\htau_{k})}\right) \cdot (\htau_k - p_{k+1})
\end{equation}
where for any probability $x \in [0,1]$, $P_n(x)$ is defined as:
\begin{equation}
    P_n(x) = \sum\limits_{i=0}^{n-1}\frac{1}{i+1}\binom{n-1}{i}x^{n-1-i}(1-x)^i = \frac{1}{n}\cdot \frac{1-x^n}{1-x}
\end{equation}

We now turn to the batched-prophet inequality. 

\begin{lemma}
\label{lem:batched-ratio}
The following sequence of thresholds $\tau_1 >\tau_2 > \dots > \tau_{\ell}$ achieves a $1 - \frac{1}{e^{\ell}}$ approximation of the optimal reward $\OPT$:
\begin{align*}
    \tau_1= G^{-1} \left( \left(\frac{1}{e}\right)^{\frac{1}{n}}\right), \tau_2= G^{-1} \left( \left(\frac{1}{e^2}\right)^{\frac{1}{n}}\right), \dots, \tau_{\ell}= G^{-1} \left( \left(\frac{1}{e^{\ell}}\right)^{\frac{1}{n}}\right).
\end{align*}
\end{lemma}

The proof, following \cite{Alaei22DescendingDescrete}, relies on constructing a sequence of thresholds that balances the tradeoff between having a high reward from the current round and having a high probability of collecting future rewards. 
For completeness, we give the whole proof in \Cref{sec: omitted_proofs}.
In order now to connect the Batch Prophet Inequality with the user's net utility maximization in $\ell$-DA, we use \eqref{eq:batch-to-rev}, and thus, we retrieve the equilibrium thresholds $\htau_1, \htau_2, \dots , \htau_{\ell}$ as follows:
\begin{align*}
    \htau_1= \bphi_{1}^{-1}(\tau_1), \; \htau_2= \bphi_{2}^{-1}(\tau_2), \; \dots, \; \htau_{\ell}= \bphi_{{\ell}}^{-1}(\tau_{\ell}).
\end{align*}
Lastly, we can transform the equilibrium thresholds $\htau_1, \htau_2, \dots , \htau_{\ell}$ to the announced prices $p_1, p_2, \dots, p_{\ell}$, by setting $p_{\ell}=\htau_{\ell}$ and then using \eqref{eq:indifference} inductively. Thus, we ended up with $\ell$ price-privacy tuples that achieve $1- \frac{1}{e^{\ell}}$ approximation of the optimal revenue, assuming that the privacy levels are  partitioned in $[0,1]$.

\section{Privacy Marketplace}
\label{sec:privacy_marketplace}

In this section, we extend the previous one-sided market of having an individual user selling her transaction under her privacy preferences to a two-sided market - a \textit{privacy marketplace} - whereby multiple users sell their transactions under their privacy preferences and \textit{unit-demand} searchers buy the right to execute them. In the middle, there is the \textit{matchmaker} who aggregates users' privacy preferences and searchers' valuations for the users' transactions. 
Importantly, in the two-sided market that we model, we assume that each user (who acts as a seller) maintains a private outside-option value for retaining her transaction. In that regard, unsold transactions retain their privacy utility and can still be executed elsewhere, capturing the opportunity cost. 
Therefore, our benchmark will be the social welfare of the market which is the sum of the searchers' welfare from the transactions bought and the users' welfare from retaining the unsold transactions.

\paragraph*{Model.}
We now define the formal structure of the two-sided market. There are $m$ users, each one selling her transaction and $n$ unit-demand searchers that want to buy users' transactions. From now on, we use the letter $i$ to denote a single searcher and the letter $j$ to denote a single user or her associated transaction. Each user $j$ has a single transaction that can be sold under different privacy preferences, which are defined by  a differential privacy parameter $\e_j \in [0,1]$ and a privacy cost function $c_j: [0,1] \to \Rplus$. As in the previous setting, for our modeling purposes, we discretize the privacy interval into $\ell$ privacy levels, namely $k=\{1,2 \dots , \ell\}$. Thus, we will write the user $j$'s privacy cost of selling her transaction at privacy level $k$, as $c_{j,k}$. 
For each user $j$, if her transaction is sold at privacy level $k$ her utility is the amount received by the matchmaker minus the incurred privacy cost $c_{j,k}$. If the transaction is not sold, user $j$ captures the opportunity cost $c_{j,k}$ as explained in the previous paragraph.  
On the other side of the market, each searcher $i$ is assumed to be unit-demand with normalized valuations, which means each searcher is interested in purchasing at most one transaction and gets no value if she is not allocated any transactions.\footnote{Searchers compete to buy a pending transaction and bundle it with their own to capture arbitrage opportunities, liquidations or "sandwiches".} 
Each searcher $i$ has a quasilinear utility for getting transaction $j$ which is equal to her private valuation for transaction $j$ at the respective privacy level minus the price paid.
As in the one-sided market studied in the previous section, we consider a Bayesian setting, where there is public knowledge of searchers' valuation distribution and users' privacy cost distributions. Thus, we consider a user $j$ of having a type: 

\begin{equation*}
    t_j = (c_{j,1}, c_{j,2}. \dots , c_{j,\ell}) 
\end{equation*}
where each coordinate $c_{j,k}$ is her privacy cost for selling her transaction at privacy level $k$. The privacy costs distribution of the users are independent and publicly known to the matchmaker, that is, 
\begin{equation*}
    G(t_1, \dots , t_m) = \prod_{j=1}^m \prod_{k=1}^{\ell}G_{j,k}(c_{j,k})
\end{equation*}
where, each marginal $G_{j,k}$ is the known CDF of user j's privacy cost of having her transaction sold at privacy level $k$.
On the other side of the market, we consider a searcher $i$ of having a type: 
\begin{equation*}
    t_i = \left(v_{i,1,1}, v_{i,1,2}, \dots, v_{i,1,\ell}; \dots ; ,\ v_{i,m,1}, \dots , v_{i,m,\ell} \right)
\end{equation*}
where each coordinate $v_{i,j,k}$ is her valuation for transaction $j$ at privacy level $k$.
The matchmaker knows the prior distribution over the profile $(t_1, \dots ,t_n)$. We assume independence across searchers and transactions/privacy levels, so that 
\begin{equation*}
    F(t_1, \dots , t_n) = \prod_{i=1}^n\prod_{j=1}^m\prod_{k=1}^{\ell}F_{i,j,k}(v_{i,j,k})
\end{equation*}
where each marginal $F_{i,j,k}$ is the known CDF of searcher $i$'s value for transaction $j$ at privacy level $k$. Thus, the privacy marketplace is defined by the tuple $(n,m,l,\mathbf{F},\mathbf{G})$, where $[n]$ denotes the set of searchers, $[m]$ denotes the set of users, $[\ell]$ denotes the privacy levels,  $\mathbf{F}$ is the vector of searchers' valuation distributions and $\mathbf{G}$ is the vector of users' privacy cost distributions. Moreover, an allocation for the privacy marketplace $(n,m,l,\mathbf{F},\mathbf{G})$ is an array 
$\allocs = (x_{i,j,k})_{(i,j,k) \in [n] \times [m] \times [\ell]} \in \{0,1\}^{[n] \times [m] \times [\ell]}$,
where $x_{i,j,k} =\{0,1\}$ denotes that searcher $i$ is allocated transaction $j$ at privacy level $k$. We remind that searchers are unit-demand, which means that for every searcher $i$, it holds $\sum\limits_{j=1}^m \sum\limits_{k=1}^{\ell} x_{i,j,k} \in  \{0,1\}$ and each transaction $j$ can be sold at exactly one privacy level, corresponding to $\sum\limits_{i=1}^n \sum\limits_{k =1}^{\ell}x_{i,j,k} \in  \{0,1\}$. We introduce the decision variable $y_{j,k} \in \{0,1\}$ that indicates transaction $j$ is unsold and kept at privacy level $k$. Then for each transaction $j$, it holds $\sum\limits_{i=1}^n\sum\limits_{k=1}^ {\ell}x_{i,j,k}+\sum\limits_{k=1}^{\ell}y_{j,k}=1$.
We are now ready to define the social welfare of a mechanism for the two-sided market under a profile $(\vals, \costs)$ and an allocation rule $\allocs$:

\begin{equation}
\label{eq:welfare_two_sided}
    \SW_{\allocs}(\vals, \costs) = \sum_{i=1}^n \sum_{j=1}^m \sum_{k=1}^{\ell} \left( v_{i,j,k} - c_{j,k} \right)\cdot x_{i,j,k} \ + \  \sum_{j=1}^m \sum_{k=1}^{\ell}
    c_{j,k} \cdot y_{j,k}
\end{equation}
The first term in \eqref{eq:welfare_two_sided} is the searchers' welfare and the second term in \eqref{eq:welfare_two_sided} is the users' welfare. Additionally, the privacy marketplace consists of a payment vector $\prices = (\prices^S, \prices^U) \in \mathbb{R}^n \times \mathbb{R}^m$, where $\prices^S$ refers to the searchers’ vector of payments and $\prices^U$ to the users’ refunds. 
We would like to design a mechanism for the privacy marketplace that apart from obtaining high social welfare, satisfies the standard desired properties.
\begin{itemize}
    \item \textbf{Dominant Strategy Incentive Compatibility (DSIC).} It is a dominant strategy for every agent (searchers and users) to sincerely report their true beliefs (valuations and privacy costs respectively). In other words, for every searcher $i$ (resp. user $j$) for every vector of valuations and privacy costs, searcher $i$ (resp. user $j$) cannot increase her utility by misreporting her belief.
    \item \textbf{Individual Rationality (IR).} It is not harmful for any agent to participate in the mechanism.
    \item \textbf{Budget Balanced (BB).} The sum of all the payments is zero. In the two-sided market, this means that the sum of the searcher payments is equal to the sum of the payments received by the users and hence no external party needs to subsidize the mechanism.\footnote{It is known from the work of Myerson and Satterhwaite \cite{Myerson83Bilateral} that it is impossible to design a social welfare maximizing mechanism while being IR, IC and BB even in the bilateral-trade setting, where there is one seller and one buyer}
\end{itemize}

\subsection{An Approximately Efficient Privacy Marketplace.}
\label{subsec:two-sided-posted-price}
The mechanism that we describe relies on \cite{Colini20approximately}.
The process of designing an efficient two-sided market will be based on known techniques from combinatorial auctions through posted prices in standard one-sided markets. In a standard one-sided static posted-price mechanism with $m$ items, the mechanism designer posts a menu of prices $(p_j)_{j \in [m]}$, then an arbitrary order of buyers $i=1, \dots , n$ is fixed and subsequently each buyer $i$ arrives sequentially and purchases her utility-maximizing bundle of unsold items. 

First, for the privacy marketplace, we focus on budget balance mechanisms. In particular, we focus on mechanisms whereby if a trade happens between searcher $i$ and user $j$, the payment of searcher $i$ is transferred to user $j$, which means $p^S_i = p^U_j$. Such mechanisms are obviously budget-balanced.

The mechanism works as follows. Initially, it fixes a uniform privacy level across all transactions, for example, $k=0$, and by considering the respective marginal CDFs $F(t_1, \dots , t_n) = \prod\limits_{i=1}^n\prod\limits_{j=1}^mF_{i,j,0}(v_{i,j,0})$ and $G(t_1, \dots , t_m) = \prod\limits_{j=1}^m G_{j,0}(c_{j,0})$, it computes the matching that maximizes the expected social welfare. 
For each matched pair $(i,j)$, which from now on, we will refer to as $(i_j^*,j)$, we want to choose the privacy level $k_j^*$ that maximizes the ex-ante expected social welfare $\SW_{j,k}$ 
\begin{equation}
\SW_{j,k} = \mathbb{E}[c_{j,k}]+\mathbbm{1}\left\{\mathbb{E} \left[v_{i_j^*,j,k}\right] \ge 4 \cdot \mathbb{E}[c_{j,k}]\right\} \cdot \frac{1}{2}\left(\mathbb{E} \left[v_{i_j^*,j,k}\right] -  \mathbb{E}[c_{j,k}]\right)   
\end{equation}
Therefore,
\begin{equation}
\label{eq:optimal-privacy-level}
k_j^* = \arg\max\limits_k \SW_{j,k}
\end{equation}
Thus, based on \eqref{eq:optimal-privacy-level}, we derive the optimal privacy vector $\mathbf{k}^* = (k_1^*, \dots , k_m^*)$. 
We now partition the transactions into two disjoint sets. A transaction is said to be \textit{worth-matching} if and only if $\mathbb{E} \left[v_{i_j^*,j,k_j^*}\right] \ge 4 \cdot \mathbb{E}[c_{j,k_j^*}]$.
Hence, we define the set $W = \left\{j : \mathbb{E} \left[v_{i_j^*,j,k_j^*}\right] \ge 4 \cdot \mathbb{E}[c_{j,k_j^*}]\right\}$
For each pair $(i,j)$ such that $j \in W$, we define the posted price $p_j$:
\begin{equation}
\label{eq:posted_prices}
    p_j = \frac{\mathbb{E} \left[v_{i_j^*,j,k_j^*}\right]}{2}
\end{equation}
Thus, for each transaction $j$, the mechanism posts a price-privacy tuple $(p_j, k_j^*)$. Subsequently, the mechanism asks each user $j$ with probability 
\begin{equation}
\label{eq:probablistic-query}
    q_j = \frac{1}{2 \cdot G_{j,k^*_j}(p_j)}
\end{equation}
whether she accepts the offer $(p_j, k^*_j)$, that is, whether for her realized privacy cost $c_{j,k^*_j}$, it holds $p_j \ge c_{j,k^*_j}$.
Then, searchers arrive sequentially (in arbitrary order) and acquire their utility-maximizing transaction. Next, we fix an arbitrary order of searchers $i=1, \dots , n$, and the mechanism asks each searcher sequentially which is the utility-maximizing transaction given the price-privacy menu among transactions that are still available.

Before proceeding with analyzing the performance of the posted price mechanism, we remind that the pricing scheme \eqref{eq:posted_prices} is based on the well-known posted price mechanism for a single item \cite{Kleinberg12Matroid-prophet}, which states that, in a Bayesian setting, posting a price $p$ equal to half of the expected highest value results in a $\frac{1}{2}$-approximation of the optimal social welfare. This can be extended to combinatorial auction auctions and in particular, in our case, for auctions with unit-demand buyers. We provide this statement below in \Cref{lem:matching-posted-price}. Moreover, we highlight, that out of all the transaction worth-trading, that is, transactions $j \in W$, we further thin the market by making an offer to corresponding user with probability $q_j$ as defined in \eqref{eq:probablistic-query}. This makes each transaction $j \in W$ to be placed in the market with probability exactly 1/2 and it will be useful for the welfare analysis of the posted-price mechanism.


\subsection{Analysis of the Mechanism.}

\begin{theorem}
The posted-price mechanism for the privacy marketplace described in \Cref{subsec:two-sided-posted-price} is IR, DSIC, and BB. 
\end{theorem}

\begin{proof}
It is a dominant strategy for every user $j$ to accept the price-privacy tuple that was offered if the offered price exceeds her privacy cost for the posted privacy level. Moreover, it is a dominant strategy for every searcher $i$ to choose the utility-maximizing transactions and the corresponding price-privacy tuples offered to her. Additionally, it is straightforward to see that the mechanism can never decrease any user's or searcher's utility, and therefore it is ex-post IR. Lastly, it is also clear to see that the mechanism is budget-balanced since for each pair $(i,j)$ the payment of searcher $i$ is transferred to user $j$.
\end{proof}

\begin{theorem}
The posted-price mechanism for the privacy marketplace described in \Cref{subsec:two-sided-posted-price} achieves a $\frac{1}{6}$-approximation of the optimal expected social welfare.
\end{theorem}
The proof will be based on the following two lemmas that bound separately the contribution of searchers and users in the total social welfare. We use $\OPT$ to denote the optimal social welfare and we use $\ALG$ to denote the social welfare of the posted price mechanism. Moreover, we split the the social welfare to the searchers and users contribution, that is $\OPT = \OPT^{\Search} + \OPT^{\User}$ and $\ALG = \ALG^{\Search} + \ALG^{\User}$. 
First, we give a lemma on the performance of static posted prices in an one-sided market with unit-demand buyers. Note that, in our two-sided market, after fixing the privacy vector $\mathbf{k}^*$ by \eqref{eq:optimal-privacy-level}, the agents' dimensionality is reduced and we end up with a two-sided market with $m$ items. Subsequently, we will use this lemma to upper bound the contribution of searchers' welfare in the optimal allocation by taking into account that \textit{worth trading} transactions are in the market with probability exactly 1/2.

\begin{lemma}
\label{lem:matching-posted-price}
Consider an auction with $n$ unit-demand bidders and $m$ different items. Let $\allocs = (x_{i,j})_{(i,j) \in [n] \times [m]} \in \{0,1\}^{[n] \times [m]}$ be a feasible matching allocation. Under a realized valuation profile $\vals$, let $x^*$ be an optimal allocation and 
$\OPT(\vals) = \sum\limits_{i=1}^n \sum\limits_{j=1}^m v_{i,j}\cdot x_{i,j}^*$ be the welfare of the optimal allocation. We define $V_j^*$ as the marginal contribution of item $j$ in the optimal matching $x^*$.\footnote{Under a matching constraint, the marginal contribution of item $j$ is just the bidder $i$'s valuation $v_{i,j}$ who is allocated item $j$.} The posted price mechanism that posts for each item $j$ the price $p_j =  \frac{\mathbb{E}_{\vals} \left[V^*_j(\vals) \right]}{2}$ achieves at least $\frac{1}{2}$-approximation of the optimal social welfare.
\end{lemma}

Note that, the optimal social welfare in the one-sided market upper bounds the searchers' contribution $OPT^{\Search}$ in the optimal social in our two-sided market since in $OPT^{\Search}$ contribute only the transactions worth-trading and each transaction $j \in W$ is eventually in the market with probability  $q_j = \frac{1}{2\,\Pr[c_{j,k_j^*} \le p_j]}$. Therefore, 
\begin{equation}
\label{eq:searchers-bound}
    \ALG^{\Search} \ge \frac{1}{4}\sum_{j \in W}\mathbb{E}[v_{i^*_j,j,k^*j}]
\end{equation}
Next, we provide a lemma for the transaction that are not worth-trading, i.e. for each matched pair $(i^*_j,j)$ such that $j \not \in W$, it holds that $4\mathbb{E}[c_{j,k_j^*}] > \mathbb{E}[v_{i^*_j,j,k^*j}]$. Hence
\begin{equation}
\label{eq:user-not-trading}
\ALG^{\User} \ge \sum_{j \not \in W}\mathbb{E}[c_{j,k_j^*}] \ge \frac{1}{4}\sum_{j \not \in W}\mathbb{E}[v_{i^*_j,j,k^*j}]
\end{equation}
Thus, by adding \eqref{eq:searchers-bound} and \eqref{eq:user-not-trading} we sum over all transactions $j \in [m]$ and we obtain:
\begin{equation}
\label{eq:searchers-upper-bound}
\ALG^{\Search} + \ALG^{\User} \ge \frac{1}{4} \OPT^{\Search}     
\end{equation}

\begin{lemma}[Users' contribution bound]
\label{lem:users-unit-demand}
If every worth-trading transaction $j \in W$ becomes available with probability exactly $1/2$, then for the users' welfare it holds
\begin{equation}
\label{eq:users-upper-bound}
2\,ALG^{U} \;\ge\; \sum_{j=1}^{m}\mathbb{E}[c_{j,k_j^*}] \;\ge\; OPT^{U},
\end{equation}
\end{lemma}

\begin{proof}
For each transaction $j\in W$, its posted price is $p_j=\tfrac12 \mathbb{E}[v_{i^*_j,j,k_j^*}]\ge 2\,\mathbb{E}[c_{j,k_j^*}]$. 
So, for each transaction $j \in W$, we have
\begin{equation*}
\Pr[c_{j,k_j^*}> p_j] \le \Pr[c_{j,k_j^*} > 2\mathbb{E}[c_{j,k_j^*}] \le \frac{1}{2}, 
\end{equation*}
where the last inequality holds due to Markov's inequality.
Hence, $\Pr[pj \ge c_{j,k_j^*}] \ge$ which means that, with probability at least $1/2$, user $j$ gains value by selling at price $p_j$. and this means that
each transaction $j\in W$ contributes at least $\mathbb{E}[c_{j,k_j^*}]/2$ in expectation to the users' side. Moreover, each $j\not\in W$ never trades and hence it contributes its full $\mathbb{E}[c_{j,k_j^*}]$.
Summing over $j$ gives $2\,ALG^{U} \ge \sum_j \mathbb{E}[c_{j,k_j^*}]$.
Clearly $\OPT^{U}\le \sum_j \mathbb{E}[c_{j,k_j^*}]$ since the users' contribution in any allocation is at most the total keep-value.
\end{proof}

\begin{proof}[Overall bound on the posted-price mechanism]
By adding \eqref{eq:searchers-upper-bound} and \eqref{eq:users-upper-bound} we retrieve the desired approximation bound:
\begin{equation}
    \OPT = \OPT^{\User} + \OPT^{\Search} \le 2\cdot \ALG^{\User} + 4 \cdot \ALG^{\User} + 4 \cdot \ALG^{\Search} \le 6 \cdot \ALG
\end{equation}

\end{proof}

\section{Discussion on Implementations and Future Work}
In this work, we defined and formally analyzed two economic mechanisms that allow users to maintain a degree of privacy in their transactions while selling their transactions by trading off revenue for privacy through payment for order flow schemes. Importantly, such mechanisms enable users to internalize the value their transactions create in the block-building process.  

These simple mechanisms provide an optimistic outlook on mechanisms that trade efficiency for privacy, allowing customizable privacy to users. Such mechanisms fall into the emerging paradigm of Programmable Privacy. This paradigm involves using various programmable cryptographic schemes with high security assurance (e.g. Zero Knowledge Proofs, Fully Homomorphic Encryption, and Multi-Party Computation). Our work builds on the emerging line of research that attempts to improve the performance of economic mechanisms by incorporating cryptographic primitives into them.

Future work includes leveraging information-theoretic tools in order to explicitly define the loss of privacy. Last but not least, it is well known that generating randomness in blockchains is computationally costly. For example, smart contracts in Proof of Stake systems have access to Verifiable Random Functions (VRFs) which can provide some amount of entropy. If users want access to VRFs, blockchains should implement a pricing mechanism in order to allocate such scarce resources efficiently. This is closely related to the recent line of work on multidimensional fee markets \cite{Diamandis23multidimensional}. 

\bibliography{bibliography}
\bibliographystyle{unsrt}

\newpage
\appendix


\section{Revenue Curves}
\label{sec: revenue_curves}
The expected revenue of the auctioneer equals the revenue and is the expected value of the second order statistic $X_{(2)}$, and thus we can write the revenue as:
\begin{equation}
\label{eq:user_revenue_spa}
    Rev(n, \e) = \int_0^{\infty} v n(n-1)(1-F_{\e}(v))\left( F_{\e}(v) \right)^{n-2}f_{\e}(v)dv
\end{equation}
\paragraph*{Expected Revenue of the Optimal Auction.} We now write the closed form formula of the maximum expected revenue as a function of the number of bidders $n$ and the privacy parameter $\e$. This is the second price auction with a reserve price at $r_{\e} = \phi_{\e}^{-1}(0)$, where $\phi_{\e}(x) = x - \frac{1- F_{\e}(x)}{f_{\e}(x)}$. Consider the order statistics $X_{(1)}>X_{(2)} > \dots > X_{(n)}$. We distinguish three cases for the revenue:
\begin{enumerate}[(i)]
    \item $r_{\e}> X_{(1)}>X_{(2)} $
    \item $X_{(1)}>r_{\e}>X_{(2)} $
    \item $X_{(1)}>X_{(2)}>r_{\e}$
\end{enumerate}
Case (i) happens with probability $F_{\e}^n(r_{\e})$ and has no revenue. Case (ii) happens with probability $n \cdot (1- F_{\e}(r_{\e}))\cdot F_{\e}^{n-1}(r_{\e})$ and has revenue $r_{\e}$. Case (iii) happens with probability  $\sum\limits_{j=2}^{n} \binom{n}{j} \cdot (1- F_{\e}(r_{\e}))^{j} \cdot F_{\e}^{n-j}(r_{\e})$. In words, for case (iii) at least two bids have to be greater than $r$. For case (iii), the revenue depends on the number $k$ of bids that are greater than $r_{\e}$. Suppose that $j$ bids are greater than $r$, where $2 \le j \le n$ and for each $j$ consider the order statistics $Z^j_{(1)}>Z^j_{(2)} > \dots > Z^j_{(j)}$ that come from the same distribution but with support $[r_{\e}, v_{max}]$. The revenue now is equal to the second order statistic,  $Z^j_{(2)}$. All in all, the expected revenue is:
\begin{equation}
\label{eq:user_revenue_optimal}
    Rev(n, \e) = n \cdot (1- F_{\e}(r_{\e}))\cdot F_{\e}^{n-1}(r_{\e}) \cdot r_{\e} + \sum\limits_{j=2}^{n} \binom{n}{j} \cdot (1- F_{\e}(r_{\e}))^{j} \cdot F_{\e}^{n-j}(r_{\e}) \cdot Z^j_{(2)}
\end{equation}

\paragraph*{Distributions $F$.} To gain some intuition and understanding, we now study how the user's revenue behaves under two different valuation distributions $F_{\e}$. Before proceeding with the analysis, from now on, we will use \eqref{eq:user_revenue_spa} as the benchmark for the optimal revenue rather than the more complex \eqref{eq:user_revenue_optimal}. The difference between \eqref{eq:user_revenue_optimal} and \eqref{eq:user_revenue_spa} becomes negligible as the number of bidders increases, as it was shown in Bulow-Klemperer Theorem \cite{Bulow96}, and moreover, the goal of this work is to mainly analyze the impact of privacy on the revenue.
In each of the following examples, the parameter $\e$ allows us to control the distribution of bidders' valuations, with higher values of $\epsilon$ corresponding to distributions that favor higher valuations.

\begin{example}
Suppose that $F$ is the exponential distribution with rate $\e$, i.e., $F_{\e}(v) = 1 - \exp(- v/{\e})$ and $f_{\e}(v) = \frac{1}{\e}\exp(-v/{\e})$
Then, the revenue of the user is:
\begin{equation*}
\label{eq:revenue_exponential}
\begin{array}{lll}
\text{Rev}(n,\e) & = & \int_0^{\infty} v \, n(n-1) \left( 1 - (1 - \exp(-v/\epsilon)) \right) \left( \frac{1}{\epsilon} \exp(-v/\epsilon) \right) \left( (1 - \exp(-v/\epsilon)) \right)^{n-2} \, dv \\
 & = & \int_0^{\infty} v \, n(n-1) \exp(-v/\epsilon) \left( \frac{1}{\epsilon} \exp(-v/\epsilon) \right) \left( (1 - \exp(-v/\epsilon)) \right)^{n-2} \, dv \\
     & = & \frac{n (n-1)}{\e} \int_0^{\infty} v \exp(-2v/\epsilon) (1 - \exp(-v/\epsilon))^{n-2} \, dv

\end{array}
\end{equation*}
\end{example}

\begin{example} 
Consider now the following cumulative distribution $F_{\e}(v)= \left(\frac{v}{v_{max}}\right)^{\e}$, for  $v \in [0, v_{max}]$, or equivalently $F_{\e}(v)= v^{\e}$, for $v \in [0,1]$
and $f_{\e}(v)= \e v^{\e-1}$. In this case, the revenue of the user is 
\begin{equation*}
\label{eq:revenue_power}
\begin{array}{lll}
\text{Rev}(n,\e) & = & \int_0^1 v \, n(n-1) (1 - v^\epsilon) \epsilon v^{\epsilon - 1} (v^\epsilon)^{n-2} \, dv \\
     & = & n (n-1) \epsilon \int_0^1 v^{\epsilon (n - 1)} - v^{\e n} \, dv \\
     & = & \frac{\e^2 n (n-1)}{\left(\e(n-1) +1\right)(\e n +1)}

\end{array}
\end{equation*}
\end{example}

\section{Omitted Proofs}
\label{sec: omitted_proofs}
\begin{proof}[Proof of \Cref{lem:batched-ratio}]
First, we define the polynomial $P_n(\tau_i)$
\begin{equation*}
    P_n(\tau_i) = \sum\limits_{l=0}^{n-1} \frac{1}{l+1}\binom{n-1}{l}G(\tau_i)^{n-1-l}(1-G(\tau_i))^l = \frac{1}{n} \cdot \frac{1- G(\tau_i)^n}{1-G(\tau_i)} = \frac{1}{n} \sum\limits_{l=1}^{n-1}G(\tau_i)^l
\end{equation*}
which is the probability a reward $V_i$ is selected conditioned on it passing the threshold $\tau_i$. We remind that a reward is chosen uniformly at random between the rewards that pass the threshold $\tau_i$. By choosing a single threshold $\tau_1$, the expected reward obtained is at least:
\begin{equation}
\label{eq:firstround}
    \left(1- G(\tau_1)^n \right) \tau_1 + P_n(G(\tau_1)) \sum_{i=1}^n \mathbb{E}[(V_i-\tau_1)^+]
\end{equation}
Moreover, from the standard prophet-inequality analysis, we derive the following useful inequality:
\begin{equation}
\label{eq:prophet_ineq}
\tau_1 + \sum_{i=1}^n \mathbb{E}[(V_i-\tau_1)^+ \ge \tau_1 +  \mathbb{E}[(V_i-\tau_1)^+ \ge \OPT
\end{equation}
Moreover, for $\tau_1 =\left( \frac{1}{e} \right)^n$, it holds that:
\begin{equation}
\label{eq:bound:min}
    \min\{(1-G(\tau_1)^n), P_n(G(\tau_1))\} \ge 1- \frac{1}{e}.
\end{equation}
It is possible to choose a threshold $\tau_i$ such that $\min \{(1-G(\tau_i)^n, P_n(G(\tau_i)\}  \ge 1 - \frac{1}{e}$, by choosing, for example,
    $G(\tau_i) = 1 - \frac{1}{n} \approx \left(\frac{1}{e}\right)^{\frac{1}{n}}$.

Now, we derive the recursive formula for the second round. In order to get to the second round, all the rewards should be below the threshold $\tau_1$ which happens with probability $G(\tau_1)^n$. We follow the same process by just replacing the probability $G(x)$ with $\frac{G(x)}{G(\tau_1)}$, since, in order to proceed to the second round, we must condition on $V_i < \tau_1$. We can now write the expected reward:
\begin{align}
\label{eq:reward}
\mathbb{E}[\ALG (\tau_1, \tau_2, \dots , \tau_k)] \ge & \left(1- G(\tau_1)^n \right) \tau_1 + P_n(G(\tau_1)) \sum_{i=1}^n \mathbb{E}[(V_i-\tau_1)^+] \nonumber \\
& + G(\tau_1)^n \mathbb{E}[\ALG (\tau_1, \tau_2, \dots , \tau_k) | V_j \le \tau_1 , j \in [n]]
\end{align}

By using the inequality \eqref{eq:bound:min}, we can further bound the expected reward in \cref{eq:reward}

\begin{align*}
    & \left(1- \frac{1}{e} \right) \left(\tau_1 + \sum_{i=1}^n \mathbb{E}[(V_i-\tau_1)^+] \right) + \frac{1}{e} \left(1- \frac{1}{e} \right) \left(\tau_2 + \sum_{i=1}^n \mathbb{E}[(V_i-\tau_2)^+] \right) + \\
    & + \left(\frac{1}{e}\right)^2 \left(1- \frac{1}{e} \right) \left(\tau_3 + \sum_{i=1}^n \mathbb{E}[(V_i-\tau_3)^+] \right) + \dots  \\
    & \ge \left(1- \frac{1}{e} \right) \mathrm{OPT} + \frac{1}{e} \left(1- \frac{1}{e} \right) \mathrm{OPT} + \left(\frac{1}{e}\right)^2 \left(1- \frac{1}{e} \right) \mathrm{OPT} + \cdots    = \mathrm{OPT} \left(1- \frac{1}{e^k} \right).
\end{align*}
This completes the proof.
\end{proof}

\end{document}